\documentclass{notices}
\usepackage{graphicx}
\usepackage{subcaption}
\usepackage{adjustbox}
\usepackage{multicol}
\usepackage{adjmulticol}
\usepackage{graphicx}
\usepackage{caption}
\usepackage[bbgreekl]{mathbbol}
\usepackage{tkz-graph}
\usepackage{subcaption}
\usepackage{tikz}
\usetikzlibrary{arrows.meta, decorations.pathreplacing, arrows,shapes,positioning}
\usepackage{float}
\usepackage{sectsty}
\usepackage{amsrefs}
\usepackage{titlesec,amsthm}
\usepackage[margin=0.95in]{geometry}

\theoremstyle{definition}

\sectionfont{\fontsize{12}{15}\selectfont}

\usepackage[multiple]{footmisc} 
\usepackage[normalem]{ulem}
\usepackage{setspace}
\usepackage{amsmath}
\usepackage{amsfonts}
\usepackage{amssymb}
\usepackage{enumerate}
\usepackage{accents}
\usepackage{amsthm}
\usepackage{breakcites}
\usepackage{graphics,epsfig,verbatim,bm,latexsym,url,amsbsy}
\usepackage{rotating}
\usepackage{mathrsfs}
\usepackage{textgreek}
\usepackage{multirow}
\usepackage{graphicx}
\usepackage{subcaption}
\usepackage{array}
\usepackage{url}

\usepackage{graphicx}

\usepackage{float}
\usepackage{tikz}
\usepackage{caption}

\usepackage{bm}
\usepackage{tikz-cd}

\usepackage{pstricks, enumerate, pst-node, pst-text, pst-plot}
\usepackage{xparse}

\usepackage{graphicx}
\graphicspath{{screenshots/}{images/}} %

\theoremstyle{definition} 
\theoremstyle{definition}

\long\def\symbolfootnote[#1]#2{\begingroup%
\def\thefootnote{\fnsymbol{footnote}}\footnote[#1]{#2}\endgroup}

\usepackage{footnote}
\makesavenoteenv{tabular}
\usepackage{fancyhdr}
\newcommand{\documenttitle}{Thesis}

\newcommand{\argmax}{\operatornamewithlimits{arg \, max}}
\renewcommand{\max}{\operatornamewithlimits{max}}
\renewcommand{\min}{\operatornamewithlimits{min}}

\newcommand{\be}{\begin{equation}}
\newcommand{\ee}{\end{equation}}
\newcommand{\bes}{\begin{equation*}}
\newcommand{\ees}{\end{equation*}}

\usepackage{amsmath}
\usepackage{epsfig,graphics}
\usepackage{etoolbox}

\usepackage{accents}

\newcommand{\rdots}{\mathinner{%
  \mkern1mu\raise1pt\hbox{.}%
  \mkern2mu\raise4pt\hbox{.}%
  \mkern2mu\raise7pt\vbox{\kern7pt\hbox{.}}\mkern1mu}}

\DeclareMathOperator{\trace}{trace}

\usepackage{etex}
\DeclareFontFamily{U}{mathx}{\hyphenchar\font45}
\DeclareFontShape{U}{mathx}{m}{n}{
      <5> <6> <7> <8> <9> <10>
      <10.95> <12> <14.4> <17.28> <20.74> <24.88>
      mathx10
      }{}
\DeclareSymbolFont{mathx}{U}{mathx}{m}{n}
\DeclareFontSubstitution{U}{mathx}{m}{n}
\DeclareMathAccent{\widecheck}{0}{mathx}{"71}
\DeclareMathAccent{\wideparen}{0}{mathx}{"75}

\usepackage{palatino}

\makeatletter
\def\@footnotecolor{purple!30!blue}
\define@key{Hyp}{footnotecolor}{%
\HyColor@HyperrefColor{#1}\@footnotecolor%
}
\patchcmd{\@footnotemark}{\hyper@linkstart{link}}{\hyper@linkstart{footnote}}{}{}
\makeatother
\hypersetup{footnotecolor=black}

\usepackage{thmtools}
\usepackage{nameref}
\usepackage[nameinlink]{cleveref}

\newcommand\nc{\newcommand}
\nc\on{\operatorname}
\theoremstyle{definition} 
\theoremstyle{definition} \newtheorem*{thmNoNum}{Theorem}
\theoremstyle{definition} \newtheorem{definition}{Definition}
\theoremstyle{definition} \theoremstyle{remark}
\theoremstyle{definition} 
\theoremstyle{definition} 
\theoremstyle{definition} \theoremstyle{plain}
\theoremstyle{definition} 
\theoremstyle{definition} 
\theoremstyle{definition} 
\theoremstyle{definition} \newtheorem{prop}{Proposition}
\theoremstyle{definition} 
\theoremstyle{definition}

\theoremstyle{definition} 
\theoremstyle{definition}

\theoremstyle{definition} 
\theoremstyle{definition} 
\theoremstyle{definition} 
\newtheorem{fact}{Fact}

\theoremstyle{definition}

\allowdisplaybreaks

\titlespacing*{\section}
{0pt}{1.5ex plus 1ex minus .2ex}{0.5ex plus .2ex}

\title{
Spectral Methods in Microeconomics
}

\author{
  Benjamin Golub
  \affil{
    The  author is a professor of economics and computer science at Northwestern. His email address is bgolub@northwestern.edu.
    }
}

\begin{document}

\maketitle

Matrices often appear in formal models of social and economic behavior, especially models involving networks. Such models are used to study subjects ranging from opinion dynamics to pollution-mitigation negotiations to the regulation of large marketplace platforms. Matrices are used to capture the focal economic structure in each case.

Spectral theory offers powerful tools for understanding matrices, and economic modelers have leveraged these tools to gain considerable insight. When special structure is present, such as nonnegativity or symmetry, more refined tools suited to this structure---such as Perron--Frobenius theory and the spectral theorem---offer additional leverage. This essay uses these unifying mathematical threads to offer an accessible tour of several important ideas in social science, assuming minimal non-mathematical background knowledge. Though the introductions to each topic are necessarily brief, the tour cites references throughout for more context.

\section*{Central Notions}

We start with a few standard definitions from the theory of nonnegative matrices. The applications that follow provide motivation and intuition.

A nonnegative $n$-by-$n$ matrix $M=[M_{ij}]$ corresponds to a weighted digraph on the nodes \([n]=\{1,2,\ldots,n\}\), where the entry \(M_{ij} > 0\) is the weight of the edge $(i,j)$ from $i$ to $j$. A digraph is strongly connected if for any two distinct nodes $i$ and $j$, there exists a directed path from $i$ to $j$. The matrix $M$ is called irreducible if this digraph is strongly connected, or equivalently, if there exists no permutation matrix $P$ such that $P^{-1}MP$ is block upper triangular.

The \emph{spectral radius} of a matrix $M$, denoted by $\rho(M)$, is the maximum modulus of its eigenvalues.

\begin{thmNoNum}[Perron--Frobenius]
Let \(n \ge 2\), and let \(M\) be an \(n\times n\) nonnegative, irreducible matrix. Then:
\begin{enumerate}
    \item \(M\) has a positive real eigenvalue \(\lambda\) equal to its spectral radius \(\rho(M)\).
    \item There are vectors \(c, r \in \mathbb{R}^n\), with all entries strictly positive, such that \(c^\top M = \lambda\,c^\top\) and \(Mr = \lambda\,r\). These are called the left and right Perron vectors.
    \item Any nonnegative left (resp., right) eigenvector of \(M\) with any positive eigenvalue is a scalar multiple of \(c^\top\) (resp. \(r\)).
\end{enumerate}
\end{thmNoNum}

In network theory, the entries of $c^\top$ are agents' (left-hand) \emph{eigenvector centralities} in the $M$ digraph. The system of equations \begin{equation} \label{eq:centrality} \lambda c_i=\sum_j c_j M_{ji}  \text{ for each $i$},\end{equation} defining the eigenvector $c^\top$ says that node $i$'s centrality is a weighted sum of others' centralities, with $c_j$ contributing in proportion to the weight of the link from $j$ to $i$ in the digraph.

Well before the models of behavior we are about to discuss were proposed, sociologists were interested in centralities $c_i$ satisfying \cref{eq:centrality} as intrinsically plausible indices of the importance, connectedness, or status of nodes in a network. The idea that ``the cool kids are the ones that the cool kids pay attention to'' is likely to be familiar from high school or elsewhere, and the centrality equation captures this fixed-point property in linear form. Nothing anchors eigenvector centralities to any external source of node values; \cref{eq:centrality} contains no constant terms. It might therefore seem possible to consistently assign centralities to satisfy the equation in many different ways. Remarkably, however, in a strongly connected digraph these relative centralities are uniquely pinned down up to a common scale factor; this is the content of (3). 

We now turn to some applications.

\section*{Social Influence}

An adage that rings true to me\footnote{The true origin is uncertain, but its first prominent print appearance attributed it to motivational speaker Jim Rohn \cite{canfield2005success}.} goes, ``You are the average of the five people you spend the most time with.'' A simple yet surprisingly illuminating model of social learning---named for the statistician Morris DeGroot \cite{degroot1974reaching}---takes this idea seriously.

Consider a set of $n \geq 2$ agents (which might be people or, in engineering applications, robotic sensors), each with an evolving \emph{opinion} concerning some quantity. The opinion of agent $i$ is an element $x_i \in V$, where $V$ is a convex subset of a finite-dimensional vector space equipped with an inner product. For instance, $x_i$ may be a number describing how good Taylor Swift's music is; a probability distribution over some set $\Omega$ describing beliefs about the future of climate change (where $\Omega$ describes various dimensions of the observable world);  a vector representation of how agent $i$ pronounces a particular word; or an estimate of the ambient temperature. In the DeGroot model, in each period, agents update their opinions by taking weighted averages of opinions from the previous period, possibly including their own. The vector\footnote{Vectors are column vectors by default.} of opinions at time $t$, denoted by $x(t) \in V^n$, evolves according to:
$$ x(t+1) = Mx(t) \quad \text{for $t=0,1,2,\ldots$}$$ The data of this process are the $n$-by-$n$ row-stochastic matrix $M$ and the vector of initial opinions $x(0) \in V^n$, with $M_{ij} \geq 0$ representing the weight agent $i$ places on agent $j$'s most recent opinion. One interpretation is that \(M_{ij} > 0\) only if \(i\) has access to \(j\)'s opinion (e.g., because \(i\) knows \(j\) personally or follows \(j\) on social media), and the magnitude of \(M_{ij}\) reflects how much \(i\) is influenced by \(j\).  We will typically dispense with the generality of an arbitrary $V$ and focus on the case $V=\mathbb{R}$ from now on.

The dynamics are simple: $x(t) = M^t x(0)$. Can agents in this model disagree forever? Under some natural conditions, the answer is no. The conditions concern the digraph associated to $M$. A digraph is called \emph{aperiodic} if the greatest common divisor of the lengths of all its directed cycles is $1$. Let $\Delta_n$ denote the set of probability distributions on $[n]$, viewed as row vectors.

\begin{fact} \label{fact:degroot_convergence} If the digraph of $M$ is strongly connected and aperiodic, opinions converge to a consensus, meaning that $\lim_{t \to \infty} x(t) = a \mathbf{1}$ for some $a \in V$, where $\mathbf{1}$ is the vector of ones. In this case, the consensus opinion $a$ is determined by the unique left eigenvector $c^\top \in \Delta_n$ of $M$  corresponding to the eigenvalue 1:
$$ a = c^\top x(0). $$
\end{fact}

Existence and uniqueness of $c$ follow from the Perron--Frobenius theorem.\footnote{Since the strictly positive vector $\bm{1}$ is a right eigenvector of $M$ with eigenvalue $1$, the theorem gives that $1$ is a largest eigenvalue of $M$ and comes with a left eigenvector $c^\top$.} 

Our characterization of consensus follows from the fact that $\lim_{t \to \infty} M^t = \mathbf{1}  c^\top$---see \cite[Ch. 8]{meyer-book} for an excellent treatment. Indeed, $c^\top$ is simply the stationary distribution of the Markov chain associated with $M$. The key intuition is that as long as disagreement remains in a strongly connected network, some agents must moderate their opinions within a bounded number of steps. The entries of $c$ can be interpreted as measures of agents' social influence, with $c_i$ representing the weight of agent $i$'s initial opinion in the long-run consensus. Whereas the one-step updating dynamics prescribe that an agent's opinion is the average of the opinions in its neighborhood,  the process ultimately leads a strongly connected network to share a consensus opinion that blends all initial opinions---with particular weights.

The weights satisfy \cref{eq:centrality}. In the DeGroot model, it is natural that the influences are eigenvector centralities: influence comes from being listened to, and  an agent is influential when that agent has influential in-neighbors. However, one need not have many connections to be highly central: even a single link from a high-centrality node can suffice.

Do large networks aggregate information well? To examine this, \cite{GolubJackson2010} considered a sequence $(M(n))_{n=1}^\infty$ of irreducible stochastic matrices, with the matrix $M(n)$ having dimensions $n$-by-$n$, as a model of a large society. Imagine that the initial opinions $x_i(0)$ in network $n$ are drawn independently from distributions with finite, positive variances and a common expectation $\mu$. If the consensus $a(n)$ converges in probability to $\mu$, large communities enjoy the so-called wisdom of crowds: no individual's noisy opinion can obstruct convergence to the truth $\mu$. Standard weak law of large numbers arguments imply that this happens if and only if $(M(n))_{n=1}^\infty$ has associated eigenvector centralities $c_i(n)$ satisfying $\max_i c_i(n) \to 0$. In this case $(M(n))_{n=1}^\infty$ is called \emph{wise}.

We can identify a simple obstruction to wisdom. A sequence $(P(n))_{n=1}^\infty$ of nonempty subsets of nodes is \emph{prominent} if there exists $\epsilon > 0$ such that, for each $n$, there is a $t=t(n)$ with
$\sum_{i \in P(n)} (M(n)^t)_{ji} \geq \epsilon$ for all $j \notin P(n)$.
Intuitively, such a sequence is one that collectively has significant influence on all other agents after some number of updating rounds, with a lower bound $\epsilon$ that does not vanish with $n$.
\begin{prop}  The sequence $(M(n))_{n=1}^\infty$ is wise if and only if there is no prominent sequence with $|P(n)| \le C$ for some finite  $C$ independent of $n$. \end{prop}
Because of the quantification over $t$ in the definition of prominence, the result is useful mainly for quickly ruling out wisdom. \cite{GolubJackson2010} give some interpretable sufficient conditions for wisdom, but sharp and interpretable conditions are not known.

Natural generalizations of the DeGroot dynamic lead to areas with exciting open questions. \cite{CCL} study a class of nonlinear operators $T : \mathbb{R}^n \to \mathbb{R}^n$ generalizing the linear action of Markov matrix $M$ in the DeGroot model. In their terminology, an operator is \emph{robust} if it is entrywise monotone and satisfies $T(x+\gamma \bm{1})=T(x)+\gamma T(\bm{1})$ for every constant $\gamma \in \mathbb{R}$. A rich convergence theory exists for such operators in general spaces, and is surveyed in \cite{CCL,lemmens2012nonlinear}. 

The generalization of the ``wisdom of crowds'' question discussed above is as follows. Consider the long-run map $T^\infty(x):=\lim_{t\to\infty}T^t(x)$, assuming the limit exists. How much can this map depend on any small set of entries of $x$? For example, suppose we take an undirected Erd\H{o}s--Renyi graph on $[n]$ with some edge probability $p(n) \gg \log(n)/n$ as a model of a connected social network. For each degree $d$, fix a symmetric function $\tau_d:\mathbb{R}^d\to\mathbb{R}$, and for each agent $i$ of degree $d_i$, let $T_i(x) = \tau_{d_i}((x_j)_{j \in N(i)})$, where $N(i)$ is the set of $i$'s neighbors. 

If the family $(\tau_d)$ is chosen so that the resulting operator $T$ is robust and has uniformly bounded derivatives, a natural conjecture is that an analogue of wisdom should hold. The intuitive reason is symmetry in large-scale positions: agents' roles are exchangeable, and there seems to be nothing favoring the emergence of globally prominent roles by accident. \cite{CCL} make remarkable progress toward this conjecture under technical assumptions on the second-largest eigenvalue of a matrix reflecting the social network. However, under important models of social networks---sparse Erd\H{o}s--Renyi graphs, stochastic block networks, random geometric graphs with link probabilities depending on spatial proximity---these technical assumptions would not hold.  Better understanding how generalizations of centrality statistics behave in large networks with nonlinear updating behavior is a wide open and exciting problem.

\section*{A Richer Centrality}

Eigenvector centralities were defined based only on $M$, the data of the weighted digraph, with no other sources of status or prestige. But sometimes such sources are relevant. Incorporating them leads to a notion of centrality that extends the eigenvector equation by adding an exogenous term. Specifically, given a positive \emph{decay parameter} $\delta < 1/\rho(M)$ and a vector $z \in \mathbb{R}^n$, the vector of $(\delta,z)$--\emph{Katz--Bonacich centralities} $k$ is defined \cite{Jackson2008} to satisfy
\begin{equation} \label{eq:kb} k^\top = \delta k^\top M  + z^\top \; \Leftrightarrow \; k^\top = z^\top (I-\delta M)^{-1}. \end{equation}
For a quick way to remember this equation, imagine high school students having external sources of status (such as mathematical ability) and their social status coming from a linear combination of this external level and the status derived from in-connections. Below we will see some economic applications that flesh out fuller foundations for this notion of centrality.

\begin{figure}[t]
\centering
\resizebox{0.65\columnwidth}{!}{\begin{tikzpicture}[x=1cm,y=1cm]
    \tikzset{
        eig edge/.style={draw=black!48, line width=0.62pt, line cap=round, line join=round},
        eig node/.style={circle, draw=black!58, fill=green!30, line width=0.48pt, inner sep=0pt},
        walk label/.style={
            font=\fontsize{5.4}{6.0}\selectfont,
            text=black!82,
            fill=white,
            fill opacity=0.94,
            text opacity=1,
            rounded corners=1pt,
            inner xsep=1.8pt,
            inner ysep=1.6pt,
            align=center
        }
    }

    \newcommand{\eignode}[4]{%
        \pgfmathsetmacro{\diam}{0.93*pow(#4,0.5)}%
        \node[eig node, minimum size=\diam cm] (#1) at (#2,#3) {};
    }

    \eignode{n4}{-1.34}{2.10}{0.694598}
    \eignode{n5}{ 0.00}{2.76}{0.881858}
    \eignode{n6}{ 1.34}{2.10}{0.694598}
    \eignode{n3}{ 0.00}{1.12}{1.000000}
    \eignode{n1}{ 0.00}{-0.18}{0.438220}
    \eignode{n2}{ 0.00}{-1.42}{0.187260}
    \eignode{n7}{ 0.00}{-2.70}{0.069118}

    \draw[eig edge] (n3) -- (n4);
    \draw[eig edge] (n3) -- (n5);
    \draw[eig edge] (n3) -- (n6);
    \draw[eig edge] (n4) -- (n5);
    \draw[eig edge] (n5) -- (n6);
    \draw[eig edge] (n3) -- (n1);
    \draw[eig edge] (n1) -- (n2);
    \draw[eig edge] (n2) -- (n7);

    \node[walk label, anchor=east] at (-1.95, 2.14) {\shortstack[c]{$(2,7,17,$\\$51,131)$}};
    \node[walk label, anchor=west] at ( 0.46, 2.93) {\shortstack[c]{$(3,8,23,$\\$62,171)$}};
    \node[walk label, anchor=west] at ( 0.46, 1.14) {\shortstack[c]{$(4,9,28,$\\$69,200)$}};
    \node[walk label, anchor=west] at ( 0.52,-0.18) {\shortstack[c]{$(2,6,12,$\\$36,84)$}};
    \node[walk label, anchor=east] at (-0.52,-1.42) {\shortstack[c]{$(2,3,8,$\\$15,44)$}};
    \node[walk label, anchor=west] at ( 0.22,-2.70) {\shortstack[c]{$(1,2,3,$\\$8,15)$}};
\end{tikzpicture}}
\caption{A seven-node undirected graph with node areas proportional to eigenvector centralities.  Tuples near nodes list the numbers of walks of lengths $1$ through $5$ from each node. The degree-$2$ nodes in the ``head'' are markedly more central than the degree-$2$ nodes on the tail.}
\label{fig:combined}
\end{figure}

The fact that $\rho(\delta M)<1$ guarantees the absolute convergence of the Neumann series expansion $k^\top = z^\top \sum_{t=0}^\infty \delta^t M^t$, which shows that $i$'s Katz--Bonacich centrality is a sum over incoming walks (entries of $(M^t)_{ji}$ for various values of $t$), weighted by $z_j$. In particular, it is uniquely defined. Moreover, as $\delta \rho(M)$ approaches $1$ from below---the boundary at which the Neumann series diverges---we have $$(1-\delta \rho(M))k(\delta) \to y c.$$ In other words, the suitably rescaled Katz--Bonacich centrality converges\footnote{Proving this is a good exercise. A quick guide: reduce to the case where $\rho(M)=1$ and $z^\top \in \Delta_n$, and then conjugate the system by a thoughtfully chosen diagonal matrix to further reduce to the case where $M$ is row-stochastic. Study the system $\widetilde{k}(\delta)^\top = \delta \widetilde{k}(\delta)^\top M  + (1-\delta)z^\top$, noting that $\widetilde{k}(\delta)=(1-\delta)k(\delta)$. Use that $\widetilde{k}(\delta) \in \Delta_n$ to show convergence to a limit, and then think about what that limit could be.} to the eigenvector centrality $c$ of $M$ times a scalar $y$ that depends on $z$. This (along with some reflection) makes it clear that the eigenvector centrality of a node is proportional to the weight of incoming walks of very long lengths, and that dependence on $z$ fades as $\delta \rho(M)$ approaches $1$. Katz--Bonacich centrality will be involved in many of the remaining social-science applications that we will encounter.

Figure \ref{fig:combined} gives an illustration of how centrality works in a small undirected graph. Take $z=\bm{1}$ and compare the degree-$2$ nodes in the dense upper pocket with those in the tail. Even though these nodes have the same ``basic'' endowment of centrality and the same number of length-$1$ walks, those in the head have higher Bonacich centralities because they are incident to more walks of every length. As $\delta$ grows, the longer walks create dramatic differences in these centralities, which converge to eigenvector centralities.

\section*{Games on Networks}

The first of these applications is a class of game-theoretic models where agents' payoffs depend on their own actions and those of their neighbors in a network, developed by \cite{Ballesteretal2006}.

Consider a set of $n$ agents (these are often also called \emph{players} in game theory), each choosing an action $x_i \geq 0$, which can be thought of as a level of investment—say, research effort in a group project. The payoff to agent $i$ is given by:
\begin{equation}
u_i(x_1, \ldots, x_n) = -\frac{1}{2} \gamma_i x_i^2 +  \left( \beta_i + \sum_{j \neq i} G_{ij} x_j \right)x_i.
\label{eq:network_game_payoff}
\end{equation}
The first term represents the cost of effort, and the parameter $\gamma_i>0$ gives the rate at which these costs scale in effort. The other terms represent benefits: $\beta_i > 0$ is an agent-specific \emph{standalone} productivity parameter, and $G_{ij}$ measures the collaborative potential of $i$ and $j$. The project gains in value proportional to the product of their contributions. The proportionality coefficient is the strength of the link between these two agents.

We now introduce an important notion that will serve as our prediction of behavior when strategic agents play this game:
\begin{definition}
A (pure strategy) \emph{Nash equilibrium} of this game is a nonnegative vector $x^* \in \mathbb{R}^n$ such that for each agent $i$,
$$
x_i^* \in \argmax_{x_i \geq 0} u_i(x_i, x_{-i}^*),
$$
where $x_{-i}^*$ denotes the actions of all agents other than $i$.
\end{definition}

The idea behind Nash equilibrium is that each agent is choosing an action (called a \emph{best response}) that maximizes its own payoff, holding other agents' actions fixed.  This notion of stability among rational maximizers is a canonical prediction in games.

Notice that agent $i$'s best-responses do not change if we divide the utility function $u_i$ by the constant $\gamma_i$, so let us do this and define $b_i = \beta_i/\gamma_i$ and $M_{ij} = G_{ij}/\gamma_i$. With that transformation, the Nash equilibrium takes a simple form.

\begin{fact} \label{fact:Nash}
If $\rho(M) < 1$, then there exists a unique Nash equilibrium given by:
\begin{equation} \label{eq:Nash}
x^* = (I - M)^{-1} b.
\end{equation}
\end{fact}

\begin{proof}
Each $x_i$ maximizes $i$'s payoff $u_i$, taking $x_{-i}$ as given. Let $B_i(x_{-i})$ denote the unique best-response action of agent $i$. Taking the derivative of $u_i$ in $x_i$, we obtain
\begin{equation}
B(x) = b + M x. \label{eq:BR}
\end{equation}
Since $\rho(M) < 1$, the matrix $I - M$ is invertible  and the claimed solution follows.
\end{proof}

The result connects the equilibrium of the game to the network structure through the matrix $(I - M)^{-1}$. Indeed, the equilibrium action of each player is its $(1,b)$--Katz--Bonacich centrality in the network $M^\top$.

The condition $\rho(M) < 1$ has a natural interpretation. It ensures that the strategic influences of agents on one another, captured by $G_{ij}$, do not overpower increasing costs of individual effort, captured by $\gamma_i$. (If they did, actions might be inclined to explode.)

In fact, this idea can be made more precise. Rather than looking for an equilibrium, we can model a process of agents strategically adjusting their behavior in response to each other. Doing this in our current context offers a baby example of the theory of \emph{learning (to play equilibria) in games}. Suppose agents start with arbitrary actions $x(0) \in \mathbb{R}^n$ and then, in each period $t=0,1,2,\ldots$, best-respond to previous period actions. Then from \cref{eq:BR} we have the dynamic $$ x(t+1) = b + Mx(t).$$ 

If $\rho(M)<1$, which we will assume unless otherwise indicated, then iterating this dynamic yields a sequence converging to the equilibrium found above. The connection becomes clearer when we write that equilibrium using its Neumann series expansion $
x^* = (I - M)^{-1} b = \sum_{t=0}^\infty M^t b.
$
The series represents the cumulative effect of strategic interactions rippling through the network. Each term $M^t b$ captures the $t$-th order effects: how the action of agent $i$ depends on the exogenous productivity parameters of neighbors at distance $t$. These indirect effects play out as agents best-respond.

If we introduce an extra parameter $\delta$ and let $M= \delta A$ for some fixed matrix $A$, then the equilibrium action of each player becomes its $(\delta,{b})$--Katz--Bonacich centrality in the network $A^\top$. As the strength of strategic effects,  parameterized by $\delta$, grows, longer walks matter in determining equilibrium actions. If these strategic effects become too strong, so that $\rho(M)>1$, then no finite solution exists because the (positive) feedback effects discussed in the previous paragraph blow up.

The connection between Nash equilibria and network centrality measures makes the game-theoretic manifestation of network analysis very clear. An agent's equilibrium action is determined not just by its immediate connections but by its position in the broader network of strategic effects.

We can describe another aspect of our game via a different
application of Bonacich centrality.
Define the \emph{total equilibrium effort} by $X^{*}=\sum_{i=1}^n x_{i}^{*}$. We compute $X^* = \bm{1}^\top (I-M)^{-1} {b}$. 
We define the \emph{keyness} of $i$ by $\kappa_i := \partial X^*/\partial b_i$.
This derivative is the amount by which an exogenous change in $b_i$ affects
the aggregate activity. We immediately observe:
\begin{fact} Agent $i$'s keyness $\kappa_i$ is $i$'s $(1,\bm{1})$--Katz--Bonacich centrality in the network $M$. \end{fact}

Despite its simplicity, this model and its close relatives have been a useful lens for examining important practical problems. 
The model sheds light on peer effects in education, examining how a student's effort level depends on the structure of social interactions in the classroom \cite{CalvoArmengolPatacchiniZenou2009}. 
 In the study of criminal networks, where there is considerable evidence of social influence, high keyness identifies criminals whose presence drives a large amount of crime \cite{LPVZ24}. And in models of industrial R\&D investment, it provides guidance on targeting public expenditures to leverage indirect spillover effects through the network of firm collaboration \cite{KonigLiuZenou2019}.

\paragraph{The Welfare Theory of Network Games}

A core insight of game theory is that outcomes arising when agents optimize individually and noncooperatively need not be collectively optimal; the prisoners' dilemma is the standard example. This raises the question of how much value is lost to this inefficiency. If a planner, such as a manager or government, could choose the agents' actions directly, how much more total utility could be achieved?

We now analyze the efficiency of the Nash equilibrium relative to a notion of a socially optimal outcome in our quadratic network game. Spectral methods first developed by \cite{bindel2015bad} turn out to be important for gaining leverage on this question. The \emph{total welfare} in the game is defined as the sum \begin{equation} V(x) = \sum_{i=1}^n u_i(x) \label{eq:welfare_game} \end{equation} of utilities across all agents.  The socially optimal outcome, denoted $x^{\mathrm{eff}}$, is the action vector that maximizes this sum.

To quantify the inefficiency of the Nash equilibrium, we examine the \emph{price of anarchy} (PoA), which measures the worst-case ratio of total welfare at the social optimum to the total welfare at the (unique) Nash equilibrium, with the network fixed and the parameter $b$ ranging over entrywise-positive vectors.

\begin{definition}
The \emph{price of anarchy} is defined as
$
\text{PoA} = \sup_{\{b \in \mathbb{R}^n:\ b_i>0\ \forall i\}} \frac{V(x^{\mathrm{eff}})}{V(x^*)}.
$
\end{definition}

We will work in the model presented above. For this analysis we make two strong assumptions, namely that the cost coefficients $\gamma_i$ in (\ref{eq:network_game_payoff}) are all equal and the induced spillover matrix $M$ is nonnegative and symmetric. These assumptions are restrictive---a point we will return to---but they do help with an illuminating characterization. We will compare the total welfare at the Nash equilibrium $x^*$ and at the socially optimal action vector $x^{\mathrm{eff}}$.
 Our main result characterizes the PoA in terms of the spectral radius of $M$.

\begin{prop}
\label{thm:poa}
Assume $2\rho(M)<1$, where $\rho(M)$ denotes the spectral radius of $M$. Then,
$
\text{PoA} = \frac{(1-\rho(M))^2}{1-2\rho(M)}.
$
\end{prop}

\begin{proof}[Sketch]

Diagonalize $M$ as $M = W \Lambda W^\top$, where $\Lambda = \text{diag}(\lambda_1, \ldots, \lambda_n)$ and $W$ is an orthogonal matrix whose columns are eigenvectors of $M$. Let $\tilde{x} = W^\top x$ and $\tilde{b} = W^\top b$. The Nash equilibrium and socially efficient actions can be written as
$
\tilde{x}_\ell^* = \frac{\tilde{b}_\ell}{1 -  \lambda_\ell}, \quad \tilde{x}_\ell^{\mathrm{eff}} = \frac{\tilde{b}_\ell}{1 - 2 \lambda_\ell}.
$ The first follows from rewriting \Cref{fact:Nash} in the diagonal basis, and the second follows by solving a very similar system of equations corresponding to the first-order conditions for maximizing $V(x)$.

The total welfare at the Nash equilibrium is
\begin{equation}
V(x^*) = \frac{1}{2} \sum_{\ell=1}^n \frac{\tilde{b}_\ell^2}{(1 -  \lambda_\ell)^2}.  \end{equation} This is obtained by plugging in the condition that all agents are best-responding, $x^*=Mx^* + b$, into agents' utility functions from \cref{eq:network_game_payoff}, noting that in the present case $\beta=b$. This shows that $u_i = \frac{1}{2}(x_i^*)^2$, so that utilitarian welfare is $\frac{1}{2} (x^*)^\top x^*$. We then plug in the above characterization $\tilde{x}_\ell^* = \frac{\tilde{b}_\ell}{1 -  \lambda_\ell}$ in the diagonalized basis. By a similar calculation,
\begin{equation}
V(x^{\mathrm{eff}}) = \frac{1}{2}\, b^\top (I-2M)^{-1} b
 \;=\; \frac{1}{2} \sum_{\ell=1}^n \frac{\tilde{b}_\ell^2}{1 - 2 \lambda_\ell}.
\end{equation}

Thus, when $b$ is proportional to an eigenvector of $M$ with eigenvalue $\lambda$, the welfare ratio is
$\frac{V(x^{\mathrm{eff}})}{V(x^*)}=\frac{(1-\lambda)^2}{1-2\lambda}$, which is increasing in $\lambda$ on $[0,1/2)$. The PoA is therefore maximized when all weight is placed on the eigenvector corresponding to $\lambda_{\max} = \rho(M)$, yielding the result.
\end{proof}

This result reveals that the inefficiency of the Nash equilibrium increases with the network's spectral radius $\rho(M)$, a measure of strategic interaction strength. As this number approaches the maximum value where the social optimum is well-defined, the PoA grows unbounded. Moreover, the instances achieving worst-case efficiency loss are those where basic incentives are proportional to the first eigenvector of the interaction network.

The PoA in network games illuminates the impact of network structure on the efficiency of decentralized outcomes. In networks with large spectral radii, individual actions have amplified effects on others, leading to greater divergence between individual incentives and social welfare. 

\paragraph{Toward Open Questions}

The most obvious open question is what can be said about the price of anarchy when we dispense with the strong assumptions that $\gamma_i$ is constant across $i$ and $M$ is symmetric. Both impose symmetries that are neither innocuous nor likely to hold in most applications. Regarding the symmetry of $M$, it is easy to think of cases where $i$'s effort makes a big difference to $j$'s incentives but this is not so when their roles are reversed. The assumption that $\gamma_i=1$ seems even more troubling; while it can be achieved by multiplying $u_i$ by a scalar, this rescaling necessarily changes how the welfare of $i$ enters the social welfare function. The techniques used above do not allow us to vary $\gamma_i$ flexibly while still characterizing properties of a \emph{fixed} social welfare function. Can we relax these restrictive assumptions and examine generalizations of the insight that the gap between equilibrium and optimal outcomes can be bounded in terms of the strength of strategic effects? 

Incomplete information opens up another set of questions to explore. We have assumed that all the parameters of agents' utility functions---$\gamma_i$, $b_i$, $G_{ij}$---are known exactly by the agents. This is unrealistic. There are standard Bayesian models of   agents reasoning correctly about each other's parameters, each other's uncertainty about those parameters, etc. The linear structure of our basic environment works nicely with incomplete information, and understanding the welfare theory of such models would be valuable.

A final direction in which there are several active research programs but also many opportunities is the study of network interventions. The inefficiency of equilibria raises the obvious question of what to do about it. Several recent papers look at interventions in network games and dynamic processes (akin to our best-response dynamic) \cite{gaitonde2020adversarial,GaleottiGolubGoyal2020,JeongShin2024}. Many of these papers work with a similar level of tractability as the one we have been analyzing, but with different applications or interpretations---e.g., models of dynamic coordination on common opinions, rather than collaborative production.

\section*{Public Goods}

The results just presented rely on assumptions imposing specific, global functional forms on the utility functions $u_i$. These assumptions are useful for illustrating certain phenomena sharply, but they are very unrealistic.

It turns out centrality theory has more fundamental manifestations, which do not rely on such assumptions. This section illustrates this and also introduces some new economic ideas, building on our analysis of welfare and efficiency above.

Consider a set of $n$ agents, each choosing an action $x_i \geq 0$. The utility of agent $i$ is given by a strictly concave\footnote{For all $x,x'$ and all $\alpha \in (0,1)$, we have $u_i(\alpha x +(1-\alpha)x')>\alpha u_i(x) + (1-\alpha)u_i(x')$.} and continuously differentiable function $u_i(x_1, \ldots, x_n)$. We make the following assumptions:
\begin{enumerate}
    \item \emph{Costly actions}: Utility is decreasing in one's own action, holding others' actions fixed: $\frac{\partial u_i}{\partial x_i} < 0$ for all $i$.
    \item \emph{External benefits}: Utility is nondecreasing in any other player's action: $\frac{\partial u_i}{\partial x_j} \geq 0$ for all $i \neq j$.
\end{enumerate} 
The first assumption means that, holding others' actions fixed, agents would always strictly prefer to reduce their own actions.\footnote{We can instantiate this with polynomials similar to those studied above, with a cost term of the form $-\gamma_i(x_i+c_i)^2$, making the costs of action large enough at $0$. This means that the only Nash equilibrium is $x=0$.} The second means that each agent's action is at worst neutral for others.  If $x_i$ represents something like mitigating environmental harms or investing in technologies that produce clean energy, a Nash equilibrium outcome can be a tragedy: actions that are unilaterally best responses yield outcomes worse---perhaps much worse---than some other available outcomes. Our main question of this section is: if agents want to improve on the non-cooperative status quo of $0$---e.g., by negotiating a like-for-like agreement in which they commit to all provide more effort---when is this possible?

To discuss improvements to the $x=0$ status quo, we introduce some welfare theory that goes beyond the simple utilitarian notion we studied via $V(x)$ above.

\begin{definition}
A nonnegative action vector $x \in \mathbb{R}^n$ is \emph{Pareto efficient} if there is no nonnegative $y \in \mathbb{R}^n$ such that $u_i(y) \geq u_i(x)$ for all $i$ and $u_j(y) > u_j(x)$ for some $j$. 
\end{definition}

Pareto inefficient outcomes are ones such that some agents can be made strictly better off without making anyone worse off. Pareto efficiency of $x$ means that there is someone who can object to the shift to any $y$ on the grounds that it makes that agent worse off (though $y$ may be, in some reasonable sense, an improvement).
It is a fact---worth convincing oneself of---that under our assumption that each $u_i$ is strictly concave, all Pareto efficient outcomes can be obtained as maximizers of $\sum_i \theta_i u_i(x)$, where $\theta_i$ are some nonnegative numbers. As the vector $\theta$ ranges over the interior of the simplex $\Delta_n$, we trace out a manifold of codimension $1$ in $ \mathbb{R}^n$ consisting of values of $x$ solving this optimization problem.

The main point of this section is that this manifold has an interesting description in spectral terms, first suggested by \cite{ghosh2008charity} and developed by \cite{elliott2019network}. Define the \emph{benefits matrix} $B(x)$ as
$$
B_{ij}(x) = \frac{{\partial u_i}/{\partial x_j}}{-{\partial u_i}/{\partial x_i}} \quad \text{for } i \neq j, \quad B_{ii}(x) = 0.
$$
Each entry $B_{ij}(x)$ represents the rate at which agent $i$ is willing to substitute a decrease in its own action for an increase in agent $j$'s action: i.e., the number of units of own action that agent $i$ would be just indifferent to giving in order to obtain one unit of $j$'s effort. We will assume that $B(x)$ is irreducible for every $x$, which means that it is impossible to partition the agents into two distinct sets, one of which does not care about the other's contributions.

We now have the following result.
\begin{prop}\label{prop:pareto}
An entrywise positive vector $x$ is Pareto efficient if and only if the spectral radius of $B(x)$ is equal to $1$. Moreover, under the normalization $\frac{\partial u_i}{\partial x_i} = -1$, the left-hand eigenvector centralities of $B(x)$ correspond to the {Pareto weights} $\theta_i$ such that $x$ maximizes $\sum_{i} \theta_i u_i(x)$.
\end{prop}

\begin{proof}[Proof sketch] If the spectral radius is greater than $1$, a Pareto improvement can be constructed in which one agent increases its action, generating benefits for others; then others ``pass forward'' some of the benefits they receive by increasing their own actions.

Fix any nonnegative vector ${x} \neq 0$ (which we will often suppress as an argument) and let  $\rho$ denote the spectral radius of ${B}({x})$. Then by the Perron--Frobenius theorem and the maintained assumptions, there is a nonnegative vector $c \neq 0$ such that ${B}c=\rho c$. Let ${D}$ be the Jacobian of $u$ in $x$ evaluated at the outcome of interest, with entries $D_{ij}$. Multiplying each row of this system by $-D_{ii}$,
$$  \sum_{j \neq i} \tfrac{\partial u_i}{\partial x_j} c_j + \rho \tfrac{\partial u_i}{\partial x_i} c_i = 0 \quad \forall i.$$
If $\rho>1$, then using the assumption $\frac{\partial u_i}{\partial x_i}<0$ we deduce
\begin{equation}   \sum_{j \neq i} \tfrac{\partial u_i}{\partial x_j} c_j + \tfrac{\partial u_i}{\partial x_i} c_i > 0 \quad \forall i,\label{pareto-inequality}\end{equation}
showing that a slight change where each $i$ increases its action by $c_i$ yields a Pareto improvement. The vector ${c}$ describes the relative magnitudes of contributions needed to achieve a Pareto improvement. The positivity of ${c}$ is key. The conditions of the Perron--Frobenius theorem guarantee the positivity of ${c}$. If $\rho<1$, we reason similarly by using the direction $-c$, where $c$ is the Perron vector of ${B}$---i.e., each $i$ slightly \emph{decreases} its action by $c_i$. Thus Pareto efficiency implies $\rho({B}({x}))=1$.

Conversely, if $\rho({B}(x))=1$ then $x$ is Pareto efficient. By Perron--Frobenius, if $\rho(B(x))=1$ there is a \emph{left}-hand eigenvector ${\theta}$ of $B(x)$, with all positive entries, satisfying  ${\theta}{B}(x) = {\theta}$. Under the $D_{ii}=-1$ normalization, this is equivalent to ${\theta} D(x)=0$, the first-order conditions for $x$ to maximize $ \sum_i \theta_i u_i(x)$. Since these conditions hold and the maximization problem is concave, it follows that $x$ is Pareto efficient. \end{proof}

The above argument also shows that whenever $x$ is Pareto efficient, the vector ${\theta}$ of left-hand eigenvector centralities of ${B}({x})$ is such that ${x}$ maximizes $ \sum_i \theta_i u_i({x})$. Intuitively,  $\theta_i = \sum_j \theta_j B_{ji}$ says $i$'s weight (proportional to the planner's disutility of that agent's costs) equals the total benefits it can confer on others, weighted by their $\theta_j$; if this were not so, the planner would want to change $x_i$. 

\paragraph{Essential agents}

Are there agents essential to negotiations in our setting and, if so, how can we identify them? Suppose a given agent may be exogenously unable to take any action other than $x_i=0$. How much does such an exclusion hurt the prospects for voluntary cooperation by the other agents?

Without agent $i$, the benefits matrix at the status quo of ${0}$ is obtained from ${B}({0})$ by setting all entries in row $i$ and column $i$ to zero. Call that matrix ${B}^{[-i]}({0})$. Its spectral radius is no greater than that of ${B}({0})$. The most dramatic case is one in which the spectral radius of ${B}({0})$ exceeds $1$ but the spectral radius of ${B}^{[-i]}({0})$ is less than $1$. Then by Proposition \ref{prop:pareto}, a Pareto improvement on ${0}$ exists when $i$ is present but not when $i$ is absent. To illustrate, consider the example in \Cref{fig:three-cycle-2}. Agent $4$, even though it confers the smallest marginal benefits, is the only essential agent. Without this agent, there are no cycles at all, so $\rho({B}^{[-4]}({0}))=0$. On the other hand, when 4 is present but any one other agent ($i \neq 4$) is absent, then there is a cycle whose edges multiply to more than $1$, and the spectral radius of ${B}^{[-i]}({0})$ exceeds $1$. Thus, a seemingly ``small'' agent can be essential to improving on the status quo when that agent completes cycles in the benefits network.

\begin{figure}[t]
\centering
\begin{tikzpicture}
[scale=0.6, every node/.style={transform shape}]
    \SetVertexNormal[Shape      = circle,
    FillColor = white,
    LineWidth  = 1.5pt]
    \SetUpEdge[lw         = 1.5pt,
    color      = black,
    labelcolor = white]

    \tikzset{node distance = 1.6in}

    \tikzset{VertexStyle/.append  style={fill}}
    \Vertex[x=0,y=0]{1}
    \Vertex[x=3,y=-4]{2}
    \Vertex[x=-3,y=-4]{3}
    \Vertex[x=0,y=-2.3]{4}
    \tikzset{EdgeStyle/.style={->}}
    \Edge[label=5](1)(2)
    \Edge[label=6](3)(2)
    \Edge[label=7](3)(1)
    \tikzset{EdgeStyle/.style={<->}}
    \Edge[label=0.5](1)(4)
    \Edge[label=0.5](4)(3)
    \Edge[label=0.5](4)(2)
\end{tikzpicture}
\caption{A benefits matrix ${B}({0})$ and its graphical depiction, in which player \#4 is essential despite providing smaller benefits than the others.}
\label{fig:three-cycle-2}
\end{figure}

\paragraph{Spectral Radius in Terms of Cycles}

Gelfand's formula articulates the role cycles play in Pareto improvements.

\begin{fact}\label{fact:cycles}\label{fact:monotonic}
If ${M}$ is a square nonnegative matrix, then $\rho({M})$ is equal to $\limsup_{t \to \infty} \trace \left({M}^t\right)^{1/t}$.

\end{fact}

For a directed, weighted graph with adjacency matrix ${M}$, the quantity $\trace \left({M}^t\right)$ measures the total weight of all closed walks of length $t$, where each walk is weighted by the product of its edge weights.

The essential agents are those present in sufficiently many high-value cycles. Relatedly, a single weak link dramatically reduces a cycle's value. Thus networks with an imbalanced structure, in which it is rare for those agents who could confer large marginal benefits on others to be the beneficiaries of others' efforts, will have a lower spectral radius and there will be less scope for cooperation.

\section*{Markets and Imperfect Measurement}

Imagine one of the models we have presented actually describes an economic situation. How can an analyst use these models without direct access to data such as the matrices $M$ or $B$? Imperfect inference of these objects is the best we can hope for. This section presents a response to this challenge in the context of a network game from the theory of firm behavior, following \cite{GaleottiGolubGoyalTalamasTamuz2026}. In addition to being a point of entry into statistical questions, this application illustrates how network game theory applies to a case with traditional economic ingredients like prices and quantities. 

The agents are \( n \) firms selling distinct goods.  Each chooses an action $x_i$, the price of its good. Quantities sold as a function of prices are $q(x)=q^0+Mx$, where $q^0 \in \mathbb{R}^n$. Here $M$ is an $n$-by-$n$ matrix satisfying $M_{ii}=-1$ for all $i$---which, as it turns out, is just a normalization. The entries of $M$ capture relationships between the goods. If $M_{ij}>0$ then $j$ is called a \emph{substitute} to $i$, since more of $i$ is consumed when $j$ becomes more expensive. If, on the other hand, $M_{ij}<0$, then $j$ is called a \emph{complement} to $i$. Under some standard microeconomic assumptions, $M$ is negative semidefinite. 

Firm $i$'s profit---which it seeks to maximize---is $ u_i(x) = q_i(x)(x_i-c_i),$ where $c_i$ is the firm's per-unit production cost. The relationships in $M$ create strategic spillovers in pricing: one firm's pricing affects others' optimization problems. Solving for the Nash equilibrium of the game yields the following analogue of \cref{eq:Nash}: when costs are perturbed by $\dot{c}$, equilibrium prices are perturbed by $$ \dot{x} = (I-M)^{-1} \dot{c}.$$ 

An economically important feature of this model is that \emph{market outcomes are Pareto inefficient}: perturbing the market outcome can yield an improvement in social welfare.\footnote{The basic idea is that when firms set prices, they do not have an incentive to focus on economic surplus that is not part of their profit. As a result, they can set ``socially inefficient'' prices, typically higher than a utilitarian planner would.} We thus introduce an authority that can influence the game by choosing an \emph{intervention}, a vector \( \sigma  \in \mathbb{R}^n\), which grants a subsidy to firms---changing  their costs by $\dot{c}=-\sigma$ per unit at an expenditure equal to the dot product $\sigma \cdot q$, where $q$ is the post-intervention quantity vector. (Negative entries in $\sigma$ are permitted, corresponding to a tax.)  The effect of the intervention on \emph{economic surplus}---a measure of welfare inclusive of effects on firms, consumers, and the authority's budget---turns out to be
\begin{equation}
V(\sigma) = -(q^0)^\top (I - M)^{-1} M \sigma, \label{eq:welfare_market}
\end{equation} where $q^0 \in \mathbb{R}^n$ is a given vector of pre-intervention quantities. This follows by some standard  calculations in the welfare theory of markets, which are spelled out in \cite{GaleottiGolubGoyalTalamasTamuz2026}; we will take this more elaborate analogue of \cref{eq:welfare_game} for granted. We will also take for granted the following fact: for generic data $(M, q^0)$, there are interventions such that $V(\sigma)>0$: there are subsidy-and-tax plans that generate more money than they cost.

However, this improvement requires information to implement; can it be done in practice?  Economic statisticians have developed tools to estimate the cross-firm demand effects comprising $M$, but these estimates are imperfect. Treating $q^0$ as known, suppose the authority observes $$ \widehat{M} = M + E,$$ where $E$ is a symmetric random matrix that represents estimation error, whose (above-diagonal) entries are i.i.d. with variance bounded by a constant. The challenge is that interventions cause ripple effects in this network game, captured by $(I-M)^{-1}$ in \cref{eq:welfare_market}. These effects can be quite complex, and it is unwise to substitute a noisy proxy for $M$ and take the results at face value. Instead, we suppose $(M,q^0)$ lies in some known set $\mathfrak{P}$ and define a robust intervention as one that, very probably, generates a non-vanishing amount of new social surplus per dollar spent.\footnote{We can also require good expected performance without changing the results.}
\begin{definition} \label{def:robust_intervention} An $\epsilon$-\emph{robust intervention rule} for $\mathfrak{P}$ is a function $\sigma(\widehat{M})$ so that, for all $(M,q^0) \in \mathfrak{P}$, we have
\[
\frac{V(\sigma(\widehat{M}))}{\sigma \cdot q^0} \geq \epsilon
\]
with probability at least $1-\epsilon$. \end{definition}

Note the randomness is only in the draw of $E$. 

We will give conditions under which robust interventions exist. Define the subspace \( \mathcal{L}(M, \mu) \subseteq \mathbb{R}^n \) as the span of eigenvectors of \( M \) corresponding to eigenvalues \( \lambda \) with \( |\lambda| \geq \mu \).
\begin{definition} The pair $(M,q^0)$ has \( (\mu, \delta) \)-\emph{recoverable structure} if the projection of \( q^0 \) onto \( \mathcal{L}(M, \mu) \)  has norm at least $\delta$.\end{definition} This condition  requires that $M$ have some eigenvalues of absolute value at least $\mu$, and that $q^0$ projects nonvanishingly onto their eigenspaces.
Then we have:
\begin{prop} \label{prop:GGGTT}
  Fix a sequence $\mu(n) \in \omega(\sqrt{n})$. If all $ (M,q^0) \in \mathfrak{P}$ have  \( (\mu(n), \delta) \)-recoverable structure for some $\delta>0$, then there is an $\epsilon >0$ so that $\epsilon$-robust intervention rules exist for all large enough $n$.  
\end{prop}

The significance of the condition on $\mu(n)$, namely that $\mu(n)/\sqrt{n} \to \infty$, is related to the fact that the operator norm of $E$ is $O(\sqrt{n})$ by a standard result on Wigner matrices. The condition then entails that all possible markets in $\mathfrak{P}$ have eigenvalues of absolute value larger than the norm of the noise. We now sketch how this is used to intervene robustly.

As in the analysis of the price of anarchy, let us diagonalize the matrix \( M \), writing
$M = W \Lambda W^\top,$ where \( W \) is orthogonal with columns $w^\ell$ and \( \Lambda \) is diagonal with entries \( \lambda_\ell \). It turns out that we can write the effect on welfare of an intervention as \begin{equation} \label{eq:V_market_spectral}
V(\sigma) = \sum_{\ell=1}^n \alpha_\ell \beta_\ell \frac{-\lambda_\ell}{1 - \lambda_\ell},
\end{equation}
where \( \sigma = \sum_{\ell=1}^n \alpha_\ell w^\ell \) and  \( q^0 = \sum_{\ell=1}^n \beta_\ell w^\ell \). 
The Davis--Kahan theorem then gives a way of controlling some terms in this sum even though we do not know $M$ or its eigenvectors $W$.
\begin{thmNoNum}[Davis--Kahan]
Let $M$ and $\widehat{M} = M + E$ be $n \times n$ symmetric matrices with eigenvalues $(\lambda_\ell)_{\ell=1}^n$ and $(\widehat{\lambda}_\ell)_{\ell=1}^n$, with each sequence weakly decreasing in absolute value. Let $W_k$ be the $n \times k$ matrix having orthonormal columns, with column $\ell$ being an eigenvector corresponding to $\lambda_\ell$, and define $\widehat{W}_k$ analogously. Then
\[
\Vert \widehat{W}_k \widehat{W}_k^\top - W_k W_k^\top \Vert_2 \leq \frac{2 \Vert E \Vert_2}{\delta_k}, 
\]
where $\Vert \cdot \Vert_2$ denotes the operator norm (largest absolute eigenvalue) and $\delta_k := \min_{i \leq k, j > k} |\lambda_i - \lambda_j|$ is called the eigenvalue gap.
\end{thmNoNum}

This result implies that if the eigenvalue gap $\delta_k$  is large compared to the noise norm $\Vert E \Vert_2$, then the subspace spanned by the leading $k$ eigenvectors of $\widehat{M}$ (columns of $\widehat{W}_k$) is close to the corresponding subspace for $M$ (spanned by the columns of $W_k$).
A simple special case occurs when the largest eigenvalue $\lambda_1$ is separated from the second-largest by a sufficiently large gap---much larger than $\sqrt{n}$ (see \Cref{fig:sampling_figure} for an example). Then the Davis--Kahan theorem simply says that the eigenvector $\widehat{w}^1$ is very close to $w^1$. This permits the authority to effectively control the first summand of \cref{eq:V_market_spectral}, while keeping all the other summands nearly zero, by choosing interventions that project only onto $w^1$. That is the essence of the ``robust intervention'' strategy.

Intuitively, the noise in $\widehat{M}$ means that many specific spillovers cannot be known precisely. But the identification of some large-eigenvalue eigenspaces permits the detection and estimation of latent patterns in the interactions that have a strong impact on demand responses $q(x)=q^0+Mx$. This turns out to be enough for designing good interventions. For illustration, \Cref{fig:sampling_figure} shows a specific demand matrix \(M\) and its absolute eigenvalues. One eigenspace stands out as much stronger in the spectral decomposition than others. Economically, it captures the global complementarities. The eigenvector $w^1$ corresponding to this leading eigenspace happens to be nonnegative. This implies that there is a robust intervention that subsidizes the production of all goods, which increases output and welfare both directly and indirectly through the complementarities. That is, making all goods cheaper makes everyone consume more of everything, which turns out to push prices down and demand further up, etc.

There is a remarkable coincidence being used here: Davis--Kahan says that eigenspaces with large  $|\lambda_\ell|$ are the ones that can be recovered through the noise. For an intervention $\sigma$ contained in such an eigenspace, the authority's expenditure on the intervention is $\alpha_\ell \beta_\ell$.  By \cref{eq:V_market_spectral}, a large $\lambda_\ell$ corresponds to a good conversion rate of this expenditure into social surplus: as $|\lambda_\ell|$ grows, the rate approaches $1$, so every dollar spent not only comes back but generates an extra dollar, for a net surplus of $1$. Thus, in \Cref{def:robust_intervention}, the ``return on investment'' $V(\sigma(\widehat{M}))/(\sigma \cdot q^0)$ is not only bounded away from $0$ but essentially as large as possible (recall \cref{eq:V_market_spectral}).

It could have been otherwise: for some notions of surplus in this setting, the most recoverable eigenspaces are \emph{least} consequential for the outcome of interest; i.e., expenditure along the recoverable spaces has almost \emph{no} impact.\footnote{This is true, in particular, of consumer surplus \cite{GaleottiGolubGoyalTalamasTamuz2026}.} Thus, the ability to design a good intervention hinges on an important interplay between the statistics and the microeconomics, and our understanding of this is mediated by the spectral perspective.

There are important subtleties that have been glossed over. For instance, the sketch has not explained why the conditions of \Cref{prop:GGGTT} rule out a $q^0$ whose projection onto leading eigenspaces is too small. Another important challenge in the general analysis is dispensing with any assumptions on eigenvalue gaps, which are not needed for \Cref{prop:GGGTT}.  The interested reader is referred to \cite{GaleottiGolubGoyalTalamasTamuz2026}.

\section*{Closing Reflections} 

Our applications are, in different ways, about something rippling through a network. These ripple effects seem daunting at first, but in each case the right kind of matrix analysis has tamed them to some extent. Statistical uncertainty complicates things again: how useful are the matrix formulas if the matrices are not known?

The previous section illustrates one payoff of reformulating economic problems in spectral terms: we can use statistical results on the recovery of eigenvalues and eigenvectors through noisy observation and sampling. Such recovery strategies constitute a rich and active area of research \cite{chen2021spectral}. Connections to economics promise new applications as well as new mathematical questions.

We have focused on the setting of market interventions to illustrate this interplay simply because that is where the existing research on the topic has been done. But exploiting the synergy between ``spectral microeconomics'' and matrix statistics offers exciting avenues in other applications. One that I would like to emphasize is the public goods model. It seems urgent to improve mechanisms for mutually beneficial collaboration---for instance when nations invest in mitigating disease spread, pollution, or climate change, or when businesses collaborate to establish standards or protect society from risks. Good economic design for improving on an inefficient status quo must deal not only with incentive issues but also with our uncertainty in measuring externalities and strategic spillovers.

\section*{Acknowledgments}
I thank Mark C. Wilson for soliciting this paper and for his encouragement. I also thank the referees for helpful comments that significantly improved the exposition, and Joseph Silverman for handling the final version. Yu-Chi Hsieh and Yann Calvó-Lopez provided exceptional assistance in the preparation and proofreading of the article.

\newpage 
\onecolumn

\begin{figure}[ht!]
    \centering
    \begin{subfigure}[t]{0.38\textwidth}
        \centering
        \adjustbox{valign=T}{\begin{minipage}{\textwidth}
            \centering
            \includegraphics[width=\textwidth]{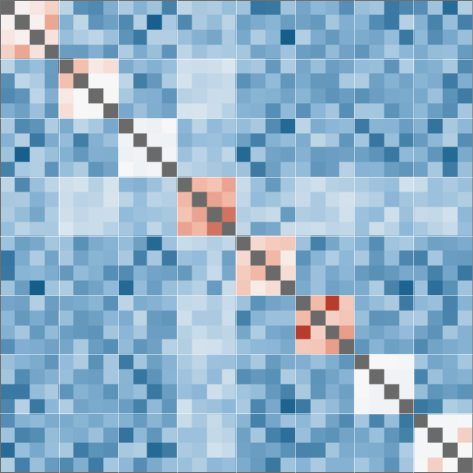}\\[4pt]
            \definecolor{dL}{RGB}{106,106,106}%
            \definecolor{sL1}{RGB}{244,207,196}%
            \definecolor{sL2}{RGB}{232,158,138}%
            \definecolor{sL3}{RGB}{210,115,98}%
            \definecolor{sL4}{RGB}{183,60,45}%
            \definecolor{cL1}{RGB}{218,230,239}%
            \definecolor{cL2}{RGB}{156,193,220}%
            \definecolor{cL3}{RGB}{118,167,201}%
            \definecolor{cL4}{RGB}{21,93,139}%
            \begin{tikzpicture}[font=\footnotesize]
              \fill[dL] (0,0.80) rectangle (0.72,1.10);
              \node[anchor=west] at (0.80,0.95) {$M_{ii}=-1$ (own-price effect)};
              \fill[cL4] (0,0.40) rectangle (0.18,0.70);
              \fill[cL3] (0.18,0.40) rectangle (0.36,0.70);
              \fill[cL2] (0.36,0.40) rectangle (0.54,0.70);
              \fill[cL1] (0.54,0.40) rectangle (0.72,0.70);
              \node[anchor=west] at (0.80,0.55) {$M_{ij}\in[-0.37,\, 0)$ (complements)};
              \fill[sL1] (0,0) rectangle (0.18,0.30);
              \fill[sL2] (0.18,0) rectangle (0.36,0.30);
              \fill[sL3] (0.36,0) rectangle (0.54,0.30);
              \fill[sL4] (0.54,0) rectangle (0.72,0.30);
              \node[anchor=west] at (0.80,0.15) {$M_{ij}\in(0,\, 0.05]$ (substitutes)};
            \end{tikzpicture}
            \vspace{.1in}
        \end{minipage}}
        \caption{An example $M$ matrix with $32$ products arranged in groups of $4$. Products are typically substitutes with others in their own group, which can be interpreted as the same type of good.  Across groups, there are complementarities with heterogeneous intensities $M_{ij}$. For a concrete example, consider different categories of computer peripherals.}
        \label{fig:dmatrix}
    \end{subfigure}
    \hfill
    \begin{subfigure}[t]{0.58\textwidth}
        \centering
        \adjustbox{valign=T}{\resizebox{\textwidth}{!}{\definecolor{accent}{RGB}{0,115,70}
\definecolor{softgray}{RGB}{175,175,175}
\begin{tikzpicture}[x=1cm,y=1.25cm,every node/.append style={scale=1.5}]
  \draw[-{Stealth[length=3mm]}, thick] (0.00,0.45) -- (10.0,0.45);
  \node[font=\normalsize, anchor=north] at (5.0,0.25) {eigenvalue index};
  \draw[-{Stealth[length=3mm]}, thick] (0.60,0.00) -- (0.60,5.25) node[above, font=\normalsize] {$|\lambda_\ell(M)|$};

  \draw[thick] (0.50,3.611) -- (0.60,3.611);
  \node[font=\small, anchor=east] at (0.48,3.611) {$5$};
  \fill[accent] (0.8500,4.5178) circle (7.5pt);
  \draw[accent!85!black, line width=1pt] (0.8500,4.5178) circle (7.5pt);
  \node[font=\normalsize\bfseries, text=accent!85!black, anchor=west] at (1.0800,4.5178) {$|\lambda_1|$};
  \begin{scope}[xshift=18pt]
  \draw[line width=0.8pt] (1.65,1.4977) -- (1.65,4.5178);
  \draw[line width=0.8pt] (1.50,4.5178) -- (1.80,4.5178);
  \draw[line width=0.8pt] (1.50,1.4977) -- (1.80,1.4977);
  \node[font=\small, anchor=west] at (1.85,3.00) {gap};
  \end{scope}
  \fill[softgray] (1.1210,1.4977) circle (2.5pt);
  \fill[softgray] (1.3919,1.4134) circle (2.5pt);
  \fill[softgray] (1.6629,1.3367) circle (2.5pt);
  \fill[softgray] (1.9339,1.3245) circle (2.5pt);
  \fill[softgray] (2.2048,1.2217) circle (2.5pt);
  \fill[softgray] (2.4758,1.2059) circle (2.5pt);
  \fill[softgray] (2.7468,1.1820) circle (2.5pt);
  \fill[softgray] (3.0177,1.1661) circle (2.5pt);
  \fill[softgray] (3.2887,1.1529) circle (2.5pt);
  \fill[softgray] (3.5597,1.1284) circle (2.5pt);
  \fill[softgray] (3.8306,1.1209) circle (2.5pt);
  \fill[softgray] (4.1016,1.1041) circle (2.5pt);
  \fill[softgray] (4.3726,1.0860) circle (2.5pt);
  \fill[softgray] (4.6435,1.0583) circle (2.5pt);
  \fill[softgray] (4.9145,1.0419) circle (2.5pt);
  \fill[softgray] (5.1855,1.0327) circle (2.5pt);
  \fill[softgray] (5.4565,0.9957) circle (2.5pt);
  \fill[softgray] (5.7274,0.9816) circle (2.5pt);
  \fill[softgray] (5.9984,0.9472) circle (2.5pt);
  \fill[softgray] (6.2694,0.9241) circle (2.5pt);
  \fill[softgray] (6.5403,0.9026) circle (2.5pt);
  \fill[softgray] (6.8113,0.8679) circle (2.5pt);
  \fill[softgray] (7.0823,0.8064) circle (2.5pt);
  \fill[softgray] (7.3532,0.7968) circle (2.5pt);
  \fill[softgray] (7.6242,0.7583) circle (2.5pt);
  \fill[softgray] (7.8952,0.6044) circle (2.5pt);
  \fill[softgray] (8.1661,0.5533) circle (2.5pt);
  \fill[softgray] (8.4371,0.5019) circle (2.5pt);
  \fill[softgray] (8.7081,0.4767) circle (2.5pt);
  \fill[softgray] (8.9790,0.4677) circle (2.5pt);
  \fill[softgray] (9.2500,0.4522) circle (2.5pt);
  \node[font=\small, text=gray!70!black] at (6.00,1.33) {small $|\lambda_\ell|$};
  \begin{scope}[shift={(4.8,3.25)}, x=0.864cm, y=0.792cm, every node/.style={scale=1.2}]
    \fill[accent!5] (-0.25,-0.58) rectangle (5.83,2.79);
    \draw[accent!20,thin] (-0.25,-0.58) rectangle (5.83,2.79);
    \node[font=\footnotesize\bfseries,accent!70!black,anchor=north west,align=left] at (-0.05,2.64) {Top-eigenvalue eigenvectors\\Are recoverable};
    \draw[-{Stealth[length=2mm]},accent,line width=1.2pt] (0.15,0.30) -- (2.1,1.10);
    \node[font=\small,accent!85!black,anchor=west] at (2.25,1.10) {$w^1$ {\scriptsize(true)}};
    \draw[-{Stealth[length=2mm]},accent,line width=0.9pt,densely dashed] (0.15,0.30) -- (2.2,0.55);
    \node[font=\small,accent!85!black,anchor=west] at (2.35,0.35) {$\widehat{w}^1$ {\scriptsize(est.)}};
    \fill[accent!50] (0.15,0.30) circle (1.4pt);
    \node[font=\small\itshape,accent!85!black,anchor=west] at (0.05,-0.15) {$\widehat{w}^1 \approx w^1$};
  \end{scope}
  \path (0,-1.31);
\end{tikzpicture}}}
        \caption{Absolute eigenvalues of \(M\). If the error matrix $E$ has entries with standard deviation equal to the magnitude of the typical entry of $M$ ($0.2$), the matrix norm $\Vert E \Vert_2$ is about $2 \cdot 0.2 \cdot \sqrt{n} \approx 2.3$; note the gap between the top eigenvalue and the next is  about twice as large even at this modest network size. If the number of product categories is scaled up keeping the same pattern of matrix entries, the largest eigenvalue grows at the rate $n$, so \Cref{prop:GGGTT} applies.}
        \label{fig:spectrum}
    \end{subfigure}
    \caption{Demand structure and its spectrum in an example.}
  \label{fig:sampling_figure}
\end{figure}

\end{document}